\documentclass[runningheads]{llncs}

\usepackage[T1]{fontenc}

\usepackage{graphicx}

\usepackage{amsmath,amsfonts,bm}

\def\eqref#1{equation~\ref{#1}}

\def\1{\bm{1}}

\DeclareMathAlphabet{\mathsfit}{\encodingdefault}{\sfdefault}{m}{sl}
\SetMathAlphabet{\mathsfit}{bold}{\encodingdefault}{\sfdefault}{bx}{n}

\usepackage{hyperref}
\usepackage{url}
\usepackage{enumitem}
\usepackage{dsfont}
\usepackage{graphicx}

\begin{document}

\title{Asymptotic Analysis of Trading Fees in CFMM}

\author{Peiyang Jin \and Clouds \and Jing Qian}
\institute{
\email{jpythu@gmail.com\\clouds.flowing@gmail.com\\qianjingnju@gmail.com} }

\maketitle

\begin{abstract}
\label{abs}
As the dominant trading mechanism in decentralized finance, Automated Market Maker has been widely studied in research. However, limited research has been done with the trading fees taken into consideration.
In this work, we study how much trading fee Liquidity Providers(LPs) can receive from arbitrage trading when the fee rate approaches zero. We give a closed-form formula for it and our result shows that the trading fees generated by arbitrage trading can fully offset the LVR loss when the price process is continuous. We further extend our conclusions to price processes with jumps and the theoretical analysis shows that jumps are the only cause of LP loss apart from market risks. Our results provide practical guidance for AMM designers as well as LPs.
\end{abstract}

\section{Introduction and Related Work}
\label{sec:intro}
In recent years, Decentralized finance (DeFi) has gained popularity in financial markets due to blockchain technology, which facilitates the trading of assets in a non-custodial manner. At the time of writing, the total value locked (TVL) in DeFi has exceeded \$100 billion. Currently, the dominant trading mechanism in DeFi is the Automated Market Maker (AMM), which is commonly deployed as a Constant Function Market Maker (CFMM). The most famous examples of CFMMs include Uniswap, Curve, and Balancer.   
Unlike traditional market makers who use order books to execute trades, CFMMs rely on mathematical functions to determine prices and liquidity. In CFMMs, Liquidity Providers (LPs) deposit assets into the pool in exchange for a share of trading fees generated by the trades in the pool, and a trade is executed only if it preserves some function of reserves, known as ''the invariant''.

With the rapid development of CFMMs, an extensive body of work has focused on analyzing the Liquidity Provider (LP) returns. \cite{evans2020liquidity} derives explicit payoff
and pricing functions for LP shares in Geometric mean market makers (G3Ms) with both static and dynamic weights.  \cite{milionis2022automated} quantifies general Constant Function Market Maker (CFMM) LP profits by Loss-Versus-Rebalancing (LVR). LVR evaluates the discrepancy between the returns of an CFMM LP position and the returns of the replicating rebalancing trading strategy, which maintains an exposure identical to the CFMM position but executes all rebalancing trades at centralized exchange (CEX) prices. \cite{fukasawa2023weighted} analyze the connections between LP Impermanent Loss in CFMMs and variance swaps, as well as gamma swaps. 

However, a question left open by the above mentioned work concerns fees. Both \cite{evans2020liquidity} and \cite{milionis2022automated} assume arbitrageurs do not pay trading fees, which is not the case in practice. In this work, we fill the gap by introducing trading fees into the model. More related to this work are \cite{evans2021optimal} and \cite{milionis2024automated}. \cite{evans2021optimal} analyze optimal fees for G3Ms and show that LP loss is minimized as fees approach zero under the assumption of a quadratic tracking error. Instead of focusing on G3Ms, this article analyzes a broader family of AMMs, CFMMs, which include G3Ms as a subclass. Besides, this paper relaxes the geometric Brownian motion assumption to a semi-martingale process. \cite{milionis2024automated} present a quantification of LP loss taking trading fees and discrete Poisson block generation times into consideration. More specifically, they focus on the expected instantaneous rate of arbitrage profit and derive a semi-closed-form expression for it. Unlike \cite{milionis2024automated}, this paper directly quantifies LP P\&L (profit and loss) by calculating the trading fees received by LP and the LP loss incurred by arbitrage without the assumption of a discrete stochastic process. 

To conclude, the main contributions of this work include:
\begin{itemize}
    \item This work presents a more accurate model for LP returns, which takes trading fees into consideration.
    \item The price process in this work is extended to a general semi-martingale, which is more generalized than previous work.
    \item This model shows that the trading fee received by LP is not zero when the fee rate approximates zero, and we derive a closed formula for it.
    \item Our results naturally extend to price processes with jumps and we evaluate the impact of price jumps on trading fee and LP's PnL. 
\end{itemize}

\section{Assumptions and Notation}
\label{sec:assumption}
An AMM is a pool of a trading pair between two assets. Without loss of generality, we assume asset A is the numeraire and asset B is a risky asset. The reserves of A and B at time $t$ are $x(t)$ and $y(t)$ respectively. For brevity, we use $x$ and $y$ to represent $x(t)$ and $y(t)$ in the following sections.

A trade against the AMM pool means that the AMM pool receives the input asset $\Delta x$ (or $\Delta y$) and outputs a controlled amount of the other asset $\Delta y$ (or $\Delta x$). Here $\Delta x,\Delta y\geq0$.
In this work, a fixed portion $\gamma$ of the input asset is collected as the trading fee and it is assumed to be a separate cash flow to LPs. More specifically, for a transaction exchanging $\Delta x$ asset A to $\Delta y$ asset B, the pool reserve changes from $(x,y)$ to $(x + (1 - \gamma) \Delta x,y - \Delta y)$ and the trading fee would be $\gamma\Delta x$ asset A. 

For a transaction exchanging $\Delta y$ asset B to $\Delta x$ asset A, the pool reserve changes from $(x,y)$ to $(x - \Delta x,y + (1 - \gamma) \Delta y)$ and the trading fee is $\gamma\Delta y$ amount of asset B. Therefore, the value of the fee changes along with the price of asset B. To simplify the calculation, a reasonable choice is to use the average price of asset B during each transaction for the calculation. Under this assumption, the value of the collected fee is  $\frac{\gamma}{1-\gamma}\Delta x$ when the input asset is $\Delta y$ asset B. 

Let $TV(f,[a,b])$ be the total variation of a function $f$ on interval $[a,b]$. The cumulative fee for a period of time $[0,T]$ is a portion of the total variation of $x$: 
\begin{equation}
\label{eq:fee}
\begin{split}
    fee(T)& =\frac{\gamma}{1-\gamma} TV(x,[0,T])\\
    & =\frac{\gamma}{1-\gamma}\sup_{\|\pi\| \downarrow0}\sup_\pi \sum_{i=1}^{n-1} |x(t_{i+1})-x(t_{i})| \\
    &(\pi:= 0\leq t_1 < t_2<...<t_n\leq T;\|\pi\|:=\max_i( t_{i+1}-t_i))
\end{split}
\end{equation} 

CFMM is a specific class of AMM where the pool reserves before and after each trade follow the equation $F(x,y)=F(x', y')$, where $F$ is a pre-defined function of a CFMM. 
In this work, we focus mainly on the CFMMs where $F\in C^2(\mathbb{R \times R})$, in which case there exists and only exists one continuous non-negative orderbook density function $g$ such that $x$ and $y$ have the representation: 
$x=\int _0 ^P g(p)\cdot pdp$, $y=\int _P ^\infty g(p)dp$. Here $P$ has the practical meaning pool price.

Denoting the price of the risky asset in the pool as $P_{\gamma}$ and the value of the pool $V=x+P_{\gamma}y$, we have the following properties:

\begin{enumerate}[label=\roman*.]
    \item $x$ and $y$ are two monotonic functions of $P_{\gamma}$
    \item $P_{\gamma}=-\frac{\mathrm{d}x}{\mathrm{d}y}$
    \item $V$ is a twice differentiable function of $P_{\gamma}$  
    \item $\frac{\mathrm{d}V}{\mathrm{d}P_{\gamma}}=y$
    \item $-\frac{\mathrm{d}^2V}{\mathrm{d}P_{\gamma}^2}=g$
\end{enumerate}
\begin{figure}
    \centering
    \includegraphics[width=0.7\linewidth]{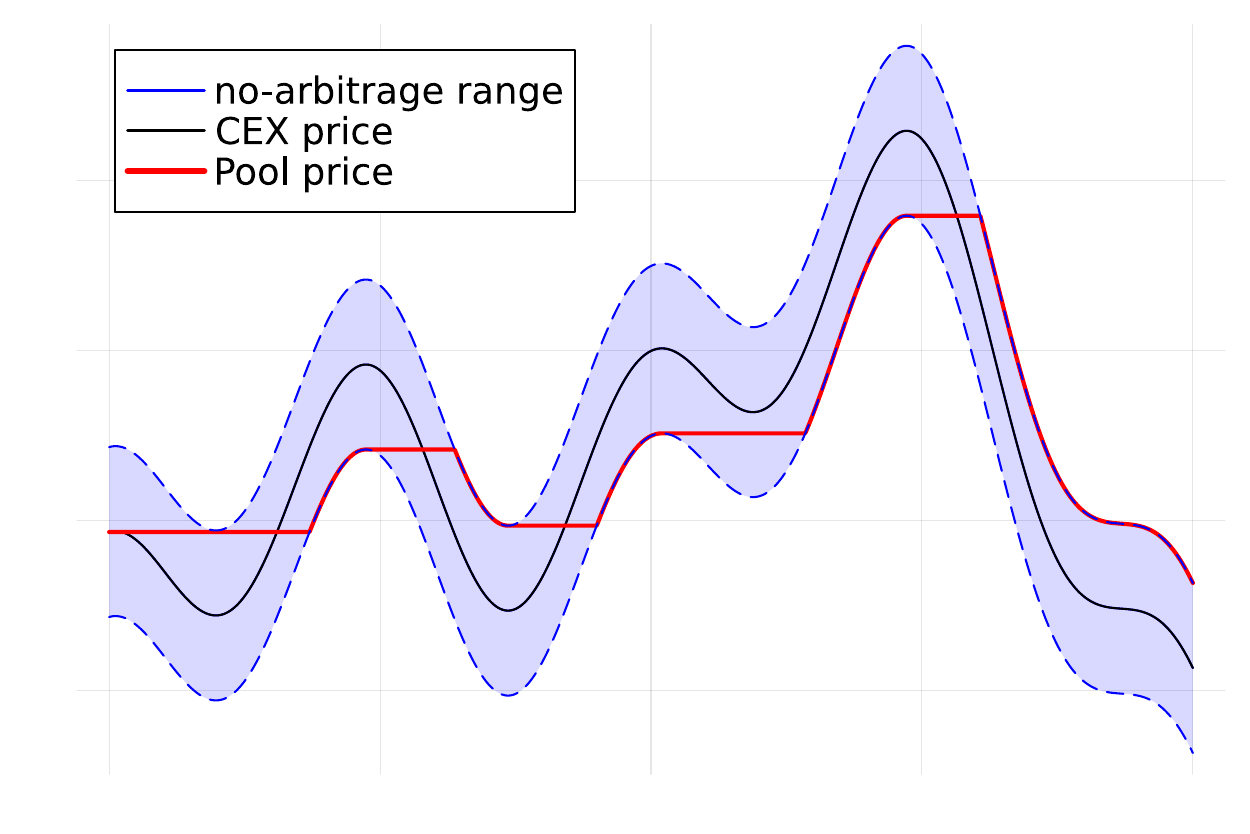}
    \caption{An illustration of CEX price change and the AMM pool price change when trading fees are present. The x axis is time. The y axis is the price of the risky asset.}
    \label{fig:price}
\end{figure}
Due to the passive nature of CFMM liquidity provision, CFMMs lags the CEX market in price change, resulting in opportunities of arbitrage. Arbitrageurs continuously monitor the price difference between the pool and the CEX and trade as long as the gap between the pool price and the CEX price becomes larger than the trading fee. These trades push the pool price close to the CEX price until the pool price reaches a ``no-arbitrage interval'', where the trading fees are larger than the profits that arbitrageurs can get by exploiting the price deviation. Figure \ref{fig:price} illustrates how CEX price and AMM pool price change over time when trading fees are taking into consideration. In this work, we assume that there are only arbitrageurs trading against the pool. In other words, all the price changes in the CFMM pool result from arbitrage activities and there are no noise trades. 

Let CEX price be a process $\{P_t\}$ and $\{P_\gamma(t)\}$ be the pool price process (respectively written as $P$ and $P_\gamma$ for brevity in the following sections). The arbitrage opportunity exists when $(1-\gamma)P_\gamma -P>0$ or $(1-\gamma)P -P_\gamma>0$.
This indicates that the pool price $P_\gamma$ changes only when $|\log P-\log P_\gamma|>-\log(1-\gamma)$. Since we assume that arbitrageurs immediately act on these opportunities and push the pool price back into the no-arbitrage interval, the pool price is a continuous adaptive process $Q$ with minimal total variation that satisfies $|\log P -\log Q|\leq-\log(1-\gamma)$. In other words, $P_\gamma$ is defined as a solution $Q$, which has
\begin{equation}
\label{eq:2}
    \inf\{TV(\log Q, [0,T]):\|\log P-\log Q\|_\infty\leq \frac{u}{2}\} ,
\end{equation}
where $u=-2\log(1-\gamma)$. 

\section{Asymptotic Analysis of Fee}
\label{sec:main}

From \cite{lochowski2013truncated}, for a càdlàg function $f$ and $c \in \mathbb{R^+}$, the truncated variation $TV^c$ of $f$ on interval $[a,b]$ is defined as

\[TV^c (f,[a,b])=\sup_n \sup_{a\leq t_1 < t_2<...<t_n\leq b} \sum_{i=1}^{n-1} \max\{|f(t_{i+1})-f(t_i)|-c,0\} ,\] 
and it has the following property:

\[TV^c (f,[a,b])= \inf_{\|g-f\|_\infty \leq \frac{c}{2}} TV(g,[a,b])\]
Applying this property to $\{\log P_t\}$ and combining with equation \ref{eq:2} lead to the following proposition.

\begin{proposition}
\label{prop:1}
\[ TV(\log P_{\gamma} ,[a,b])=TV^u(\log P,[a,b])\]

\end{proposition}

From \cite{lochowski2013truncated},  if $\{X_t\}$ is a continuous semi-martingale, $\lim_{c \downarrow 0} cTV^c (X,[0,t]) = <X>_t \ a.s.$, and the convergence is uniform for $\forall t \in [0,T]$. Combining it with proposition~\ref{prop:1}, we have
\begin{proposition}
\label{prop:2}
If $\{P_t\}$ is a continuous semi-martingale,
\[ \lim_{u \downarrow 0} uTV (\log P_{\gamma},[0,t]) = <\log P>_t \quad \ a.s.\]
The convergence is uniform for $\forall t \in [0,T]$.
\end{proposition}

\begin{theorem}
For a continuous semi-martingale process of CEX price $\{P_t\}$, 
\[\lim_{\gamma \to 0} fee(t) = \int_{0}^{t} \frac{g(P)}{2} d<P> \quad a.s.\] The convergence is uniform for $\forall t \in [0,T]$.
\label{thm:1}
\end{theorem}

\begin{proof}
Let $z=\log P_{\gamma}$, we have $\frac{dx(P_\gamma)}{dz}=x'(P_{\gamma})P_{\gamma}=g(P_{\gamma})P^2_{\gamma} \geq 0$.

According to Mean Value Theorem and equation~\ref{eq:fee}, there exists $\xi_i$ between $P_\gamma(t_i)$ and $P_\gamma(t_{i+1})$ such that
\begin{equation}
\begin{split}
    fee(t) &= \frac{\gamma}{1-\gamma} \sup_{\|\pi\| \downarrow0}\sup_\pi \sum_{i=1}^{n-1} g(\xi_i)\xi^2_i |z(t_{i+1})-z(t_i)| \\
           &=\frac{\gamma}{1-\gamma} \int_0^t g(P_\gamma(s))P^2_\gamma (s)|dz|
\end{split}
\label{eq:3}
\end{equation}

Let $h(s)=g(P(s))P^2(s)$. $h$ is bounded and uniformly continuous on $[0,t]$. Thus we can use step function to approximate $h$ uniformly. That is, $\forall \epsilon_*>0, \exists \delta>0$ such that as long as $|t_{i+1}-t_i|<\delta$, $|h(t_{i+1})-h(t_i)|<\epsilon_*$. Let $N(t)=\lfloor \frac{t}{\delta}\rfloor+1$ and $I_i=[\frac{(i-1)t}{N(t)},\frac{it}{N(t)})$. Let $\phi(s)=\sum_{i=1}^{N(t)} \mathds{1}_{s\in I_i}h(\frac{it}{N(t)})$ and $\phi(0)=h(0)$. Then we have $\|\phi-h\|_\infty<\epsilon_*$ on $[0,t]$.
Denote $J(t)=<\log P>_t$, and let $|h|\leq M(t)$ on $[0,t]$, we get the following result:
\begin{equation}
    \begin{split}
        &\left| u\int_0^t h(s)|dz| - \int_0^t h(s)d<\log P>\right|\\
        \leq &\left| u\int_0^t h(s)|dz|-u\int_0^t \phi(s)|dz|\right|+\left|u\int_0^t \phi(s)|dz|-\int_0^t\phi(s) d<\log P>\right|\\
        &+\left| \int_0^t\phi(s) d<\log P>-\int_0^t h(s)d<\log P>\right| \\ 
        \leq &\epsilon_*\int_0\!u|dz|\!+\!\sum_{i=1}^{N(t)} h(\frac{it}{N})(\left|\!\int_{I_i}u|dz|-\int_{I_i}d\!<\!\log P\!>\!\right|)+\epsilon_*\int_0^t\!d<\!\log P\!>
    \end{split}
\label{eq:4}
\end{equation}

According to proposition~\ref{prop:2}, $\forall \epsilon>0$,$\exists u_0>0$ such that $\forall u<u_0$ and $t\in(0,T]$, $\left|uTV (\log P_{\gamma},[0,t])-<\log P>_t\right|<\epsilon$. Therefore, $\int_0^tu|dz|\leq J(t)+\epsilon$. Similarly, $\left|uTV (\log P_{\gamma},[a,b])-\int_a^bd<\log P>\right|<2\epsilon$, leading to $\left|\int_{I_i}u|dz|-\int_{I_i}d<\log P>\right|\leq 2N(t)\epsilon$. Plugging these results into inequality~\ref{eq:4},

\begin{equation}
        \left|u\int_0^t h(s)|dz| - \int_0^t h(s)d<\log P>\right|<(J(t)+\epsilon)\epsilon_*+2M(t)N(t)\epsilon+J(t)\epsilon_*
\label{eq:5}
\end{equation}

Note that $g(e^z)e^{2z}$ is a bounded and uniformly continuous function of $z$ on a closed interval. Since we have $\|\log P_\gamma-\log P\|_\infty \leq\frac{u}{2}$, there exists $u_1>0$ such that as long as $u<u_1$, $|g(P_\gamma(s))P^2_\gamma (s)-h(s)|<\epsilon$ holds for all $s\in[0,t]$. Let $u<\min(u_0,u_1)$,we have \[\left|u\int_0^t g(P_\gamma(s))P^2_\gamma (s)|dz|-u\int_0^t h(s)|dz|\right|<\epsilon (J(t)+\epsilon)\]
Combining with~\ref{eq:5}, 

\begin{equation*}
\begin{split}
    &\left|u\int_0^t g(P_\gamma(s))P^2_\gamma (s)|dz| - \int_0^t h(s)d<\log P>\right|\\
    \leq & (J(t)+\epsilon)(\epsilon_*+\epsilon)+2M(t)N(t)\epsilon+J(t)\epsilon_*
    \end{split}
\end{equation*}

$J,M,N$ are non-decreasing function on $(0,T]$ and attain their maximum at $T$. Let $\epsilon_* \rightarrow0$, and then $\epsilon \rightarrow0$. We have 
\begin{equation*}
    \lim_{u \downarrow0} \left|u\int_0^t g(P_\gamma(s))P^2_\gamma (s)|dz|-\int_0^t g(P(s))P^2(s)d<\log P>\right|=0
\end{equation*}
The convergence is uniform for $\forall t \in [0,T]$.

Recall that $u=-2\log(1-\gamma)$, $\lim_{\gamma \to 0} \frac{\gamma}{u(1-\gamma)}=\frac{1}{2}$. Combining these results with equation~\ref{eq:3}, we have 
\begin{equation*}
\begin{split}
       \lim_{\gamma \to 0} fee(t) &= \lim_{\gamma \to 0} \frac{\gamma}{u(1-\gamma)}u\int_0^t g(P_\gamma(s))P^2_\gamma (s)|dz| \\
    &=\frac{1}{2}\int_{0}^{t} g(P)P^2 d<\log P> \\
    &=\int_{0}^{t} \frac{g(P)}{2} d<P> \quad \quad \quad \qed
\end{split}
\end{equation*}
\end{proof}

Theorem \ref{thm:1} shows that when the fee rate approaches zero, the accumulated trading fee is only determined by the quadratic variation of the price process, which means that for a diffusion process $dP_t=\mu dt+\sigma dW_t$, the fee is only determined by volatility $\sigma$.

Take Uniswap as a concrete example of Theorem~\ref{thm:1}.
Given range $[P_a,P_b]$ and liquidity $L$, the trading function is 
\[(x+L\sqrt{P_a})(y+\frac{L}{\sqrt{P_b}})=L^2\]
when the price $P$ is within this range. We can get
\[x(P)=L(\sqrt{P}-\sqrt{P_a}),y(P)=L(1/\sqrt{P}-1/\sqrt{P_b})\]
and 
\[g(P)=\frac{L}{2}P^{-\frac{3}{2}}\]
Applying Theorem ~\ref{thm:1} leads to the result
\[fee_t=\int_0^t\frac{L}{4}P_s^{-\frac{3}{2}}\mathds{1}_{P_s\in[P_a,P_b]}d<P>\]

Uniswap V2 is a special case of V3 when we choose $P_a=0$ and $P_b=+\infty$.
Under the assumption that the price process follows a geometric Brownian motion $dP=\mu Pdt+\sigma PdW$,
we get \[fee_t=\int_0^t\frac{L}{4}\sigma^2\sqrt{P_s}ds\]
Since the pool value $V=x+Py=2L\sqrt{P}$, the instantaneous fee ratio gets to $\frac{\sigma^2}{8}$, which is consistent with previous research \cite{evans2021optimal} and \cite{milionis2024automated}.

\section{Discussion}
\label{sec:discussion}

\subsection{LP's P\&L}

Following the definition in \cite{milionis2022automated}, LP's PnL can be decomposed to accumulated fees plus the change in pool reserve value.

\begin{equation}
\label{eq:7}
        \text{LP PnL}_{t}=\underbrace{fee_{t}}_{\text{accumulated fees}}+\underbrace{V_{t}-V_{0}}_{\substack{\text{change in pool}\\
                        \text{reserve value}
                }
        }
\end{equation}

As the fee rate $\gamma$ approaches zero, we have $\lim_{\gamma \downarrow 0}V=x+Py$, and under the assumption of a continuous semi-martingale $\{P_t\}$, we have

\begin{equation}
\label{eq:6}
\begin{split}
    \lim_{\gamma \downarrow 0}(V_t-V_0)= &\int_0^t V'(P_s)dP_s+\frac{1}{2}\int_0^t V''(P_s)d<P> \text{(from Ito's Lemma)} \\
    &=\int_0^t y(P_s)dP_s-\int_0^t\frac{g(P_s)}{2}d<P> \text{(from Section~\ref{sec:assumption})}
\end{split}
\end{equation}

Combined with Theorem ~\ref{thm:1}, we get

\begin{proposition}
\begin{equation}
    \lim_{\gamma \downarrow 0} \textup{LP PnL}_{t} = \underbrace{\int_0^ty(P_s)dP_s}_{\textup{market risk}} \quad a.s.
\end{equation}
\end{proposition}

This proposition implies that when the fee rate approaches zero and only arbitraging tradings happen, LP's PnL is only affected by the market risk. If LPs perfectly hedge the market risk at CEXs, no matter which process the price is, they can certainly(not only by expectation) make profit from non-arbitraging trading activities. 

\subsection{Process with Jumps}
In practice, AMMs are operated on blockchain networks characterized by fixed average block intervals, which impose discrete settlement intervals on AMM transactions. Consequently, even if CEX price changes continuously, AMMs price update step-wisely and exhibit 'jumpy' behavior. The longer the block time, the more pronounced these price discontinuities become. Furthermore, exogenous shocks such as breaking news and massive liquidations introduce extreme volatility, causing CEX prices to deviate sharply. To model such cases more precisely, we need to incorporate jump components into the underlying price process.

For a general semimartingale price process $\{P_t\}$, jumps are written as $\Delta P=P_t-P_{t^-}$ where $P_{t^-}$ is the left limit of $P_t$. We have the following conclusion:

\begin{proposition}
\begin{equation}
\begin{split}
    \lim_{\gamma \downarrow 0} \mathrm{\textup{LP PnL}}_{t}=\underbrace{\int_0^ty(P_s)dP_s}_{\text{market risk}} + \sum_{0\leq s \leq t} \underbrace{[V_s-V_{s^{-}}-y(P_{s^{-}})\Delta P_s]}_{\text{arbitrage loss}}\quad a.s.
\end{split}
\end{equation}
\label{proposition:4}
\end{proposition}

We then prove this conclusion formally. Recall equation~\ref{eq:7}, LP's PnL can be decomposed to the sum of accumulated fees $fee_t$ and value change of pool reserve $V_t-V_0$. We calculate these two parts separately.

To calculate $fee_t$, from \cite{lochowski2019quadratic} we have 
\[\lim_{c \downarrow 0} cTV^c (X,[0,t]) = [X]^{cont}_t \ a.s.\] holds for all càdlàg semimartingales $X$. Here $[X]^{cont}$ denotes the continuous part of the quadratic variation of $X$. Following the same proof in Section~\ref{sec:main}, we get the following corollary for the price process with jumps:

\begin{corollary}
For a general càdlàg semi-martingale process of CEX price $\{P_t\}$, 
\[\lim_{\gamma \to 0} fee(t) = \int_{0}^{t} \frac{g(P)}{2} d[P]^{cont} \quad a.s.\] The convergence is uniform for $\forall t \in [0,T]$.
\end{corollary}

To calculate $V_t-V_0$, recall~\ref{eq:6} and apply Ito's Lemma for jump-diffusion process, we have the following equation:
\begin{equation}
\begin{split}
     \lim_{\gamma \downarrow 0}(V_t-V_0)=&\int_0^t V'(P_s)dP_s+\frac{1}{2}\int_0^t V''(P_s)d[P_s]^{cont} \\
     +&\sum_{0\leq s \leq t} [V_s-V_{s^{-}}-V'(P_{s^{-}})\Delta P_s] \\
     =&\int_0^ty(P_s)dP_s+\sum_{0\leq s \leq t} [V_s-V_{s^{-}}-y(P_{s^{-}})\Delta P_s] \\
     -&\int_0^t \frac{g(P_s)}{2}d[P_s]^{cont}
\end{split}
\end{equation}

Combining the above two parts leads to~\ref{proposition:4}.

We have a brief explanation why the second part $[V_s-V_{s^{-}}-y(P_{s^{-}})\Delta P_s]$ is called arbitrage loss in~\ref{proposition:4}. Intuitively, there is an instantaneous price gap between CEX price and the AMM pool price when jump happens. As the fee rate approaches zero, LP suffers from arbitrage loss without any fee for compensation. If we take a look at LP's PnL at each jump, assuming the CEX price changes from $P$ to $P+\Delta P$, the arbitragers can make profit by pushing the pool price from $P$ to $P+\Delta P$. The amount is calculated as follows:
\begin{equation*}
\begin{split}
 \textup{ArbLoss}=&x(P+\Delta P)-x(P)-(y(P)-y(P+\Delta P))*(P+\Delta P) \\
        =&V(P+\Delta P)-V(P)-y(P)\Delta P \\
        \leq & 0 \quad \text{(Given $V$ is concave, see Section~\ref{sec:assumption})}
\end{split}
\end{equation*}

This implies that when market risk is perfectly hedged, jumps are the main cause of LP's loss. The calculation above shows that larger price jumps, which can be caused by large block time or large price volatility, lead to more LP loss. This result is consistent with \cite{milionis2024automated}.

\section{Conclusion}
In this work, we study AMM with the presence of trading fees. We derive a closed-form formula when the fee rate approaches zero. Our work relaxes the assumption that the price process is a general semimartingale. Our result shows that when the price changes continuously, the trading fees can offset the LVR loss. However, when the price jumps, LPs suffer from the loss from arbitrage trading. These results indicate that AMM designers should take blockchain latency and price jumps into consideration. Once these two issues are addressed, designers can shift their focus to other performance metrics, such as reducing the variance of LVR, without the concern that the mechanism's design might inherently impose losses on LPs. 

\bibliographystyle{splncs04}
\bibliography{amm_fee}
\end{document}